\documentclass[]{article}
\usepackage{amsmath,amssymb,amsthm}
\usepackage{mathtools}
\usepackage{xspace}
\usepackage{tikz}
\usetikzlibrary{snakes}
\usepackage{enumerate}

\newtheorem{theorem}{Theorem}
\newtheorem{problem}{Problem}
\newtheorem{definition}{Definition}
\newtheorem{example}{Example}
\newtheorem{lemma}{Lemma}

\newcommand{\Oh}{\mathcal{O}}
\newcommand{\eps}{\varepsilon}
\newcommand{\norm}[1]{\left\lVert#1\right\rVert_2}
\newcommand{\expect}[1]{\mathbb{E} \left[#1\right]}
\newcommand{\var}[1]{\mathrm {Var} \left[#1\right]}
\newcommand{\prob}[1]{\mathrm{Pr} \left[#1\right]}

\title{Approximate Hamming distance in a stream}

\author{Rapha{\"e}l Clifford, Tatiana Starikovskaya}

\begin{document}
\maketitle

\begin{abstract}
We consider the problem of computing a $(1+\eps)$-approximation of the Hamming distance between a pattern of length $n$ and successive substrings of a stream.  We first look at the one-way randomised communication complexity of this problem, giving Alice the first half of the stream and Bob the second half.    We show the following:
\begin{itemize}
 \item If Alice and Bob both share the pattern then there is an $\Oh (\eps^{-4} \log^2 n)$ bit randomised one-way communication protocol.
 \item If only Alice has the pattern then there is an  $\Oh (\eps^{-2}\sqrt{n}\log n)$ bit randomised one-way communication protocol.
\end{itemize}
We then go on to develop  small space streaming algorithms for $(1+\eps)$-approximate Hamming distance which give worst case running time guarantees per arriving symbol.
\begin{itemize}
 \item For binary input alphabets there is an $\Oh(\eps^{-3} \sqrt{n} \log^{2} n)$ space and $\Oh(\eps^{-2} \log{n})$ time streaming $(1+\eps)$-approximate Hamming distance algorithm. 
 \item For general input alphabets there is an $\Oh(\eps^{-5} \sqrt{n} \log^{4} n)$ space and $\Oh(\eps^{-4} \log^3 {n})$ time streaming $(1+\eps)$-approximate Hamming distance algorithm. 
\end{itemize}
\end{abstract}

\section{Introduction}
\label{sec:intro}
	We study the complexity of one of the most basic problems in pattern matching, that of approximating the Hamming distance.  Given a pattern $P$ of length $n$  the task is to output a $(1+\eps)$-approximation of the Hamming distance between $P$ and every $n$-length substring of a longer text. We provide the first efficient one-way randomised communication protocols as well as a new, fast and space efficient streaming algorithm for this problem.  

The general task of efficiently computing the Hamming distances offline between a pattern and a text has been studied for many years.  When the input is binary and the text has length proportional to that of the pattern, then all outputs can be computed exactly in $\Oh (n\log{n})$ time by repeated application of the fast Fourier transform~\cite{FP:1974}. For larger alphabets, $\Oh (n\sqrt{n\log{n}})$ time solutions were first discovered in the 1980s~\cite{Abrahamson:1987,Kosaraju:1987}. The fastest randomised algorithm for $(1+\eps)$-approximate Hamming distance computation for large alphabets was due for many years to Karloff from 1993~\cite{Karloff:1993} running in $\Oh (\eps^{-2} n\log^2{n})$ time overall.  In a breakthrough paper in 2015 a new algorithm was given improving the time complexity to $\Oh (\eps^{-1} n\log^3{n}\log{\eps^{-1}} )$~\cite{KP:2015}.  These fast methods all require linear space and up until this point no sublinear space solutions have been known.

The first basic question that arises is whether it is in fact possible to give a $(1+\eps)$-approximation to the Hamming distance in a streaming setting while using only sublinear space. In order to explore this question we start our study by considering two natural communication complexity problems which may also be of independent interest. Any lower bound for these communication problems will give a lower bound for the space usage of a corresponding streaming algorithm. This follows from a standard reduction where a space efficient streaming algorithm is converted into a communication protocol by taking a snapshot of memory after some symbol of the input has been read in and then sending that this snapshot to the other player. On the other hand, the upper bounds we provide will set targets for space bounds in the streaming setting.

Any streaming pattern matching algorithm using a pattern of length $n$ can be reduced to repeated application of a streaming algorithm that runs on texts of length $2n$. This is done by splitting the stream into substreams of length $2n$ which overlap by $n$ symbols.  As a result we consider communication problems with these parameter settings for pattern and text length.

\begin{problem}\label{prob:pattern-known}
Consider a text $T$ of length $2n$ and a pattern $P$ of length~$n$. Let Alice hold the information about the first half of the text and the whole of the pattern, and let Bob hold the information about the second half of the text and the pattern. Bob must output $(1+\eps)$-approximations of the Hamming distance for each alignment of  $P$ and $T$.
\end{problem}

A lower bound for the communication complexity of this problem follows from a combination of the lower bound for the communication complexity of a windowed counting problem introduced by Datar et al. in 2002~\cite{DGIM:02} and the one-way communication complexity lower bound for approximating the Hamming distance between two $n$-bit strings from~\cite{JW:2013}. 

For the first part consider the following communication problem. Assume that there is a bit vector $B$ of length $2n$.  Let Alice hold the information about the first half of $B$, and let Bob hold the information about the second half of $B$. Bob must output $(1+\eps)$-approximations of the number of set bits in each window of length $n$. Datar et al. showed that 
Alice will have to send to Bob $\Omega (\eps^{-1} \log^2 \eps^{-1} n)$ bits of information. There is a straightforward reduction from this basic counting problem to Problem~\ref{prob:pattern-known} which then gives us the same lower bound. We set $T = B$ and $P = 00 \ldots 0$ and then a $(1+\eps)$-approximation of the Hamming distance at an alignment $i$ of  $P$ and $T$ gives  a $(1+\eps)$-approximation of the number of set bits in the window $T[i, i+ n-1]$. For the second part we use Theorem 4.1 from~\cite{JW:2013} which states that the one-way communication complexity of $(1+\eps)$-approximate Hamming distance for two strings of length $n$ is $\Omega(\eps^{-2} \log{n})$ for constant error probability. 
Combining these two lower bounds together we get a lower bound of $\Omega(\eps^{-2}\log{n} + \eps^{-1}\log^2{\eps^{-1} n})$ for the communication complexity of Problem~\ref{prob:pattern-known}.  

Our first result is an efficient one-way communication protocol for Problem~\ref{prob:pattern-known} whose complexity is only slightly higher than this lower bound.  In our protocol Alice uses the fact that Bob knows the pattern as well to give an efficient encoding for parts of her half of the text which are at small Hamming distance from the pattern.

\begin{theorem}\label{th:pattern-known}
Problem~\ref{prob:pattern-known} has one-way randomised communication complexity $\Oh (\eps^{-4}\log^2 n)$.
\end{theorem}

As a model for streaming pattern matching, this communication upper bound requires that a copy of the pattern is available at all times.  Our main interest is however in algorithms whose total space complexity is sublinear in the pattern size. In order to model this situation more accurately we now consider a stronger three party communication problem.  

\begin{problem}\label{prob:pattern-unknown}
Assume that there is a text $T$ of length $2n$ and a pattern $P$ of length~$n$. Let Alice hold the information about the pattern, let Bob hold the information about the first half of the text, and let Charlie hold the information about the second half of the text. Alice will send one message to Bob who will then send one message to Charlie. Charlie must output $(1+\eps)$-approximations of the Hamming distance for each alignment of $P$ and $T$.
\end{problem}

Somewhat surprisingly, we are still able to obtain a sublinear protocol for this new problem although the bound is higher than for the simpler Problem~\ref{prob:pattern-known}.  The main technical elements of this communication protocol combine the newly introduced idea of approximate periods with succinct run-length encoded representions of the input.

\begin{theorem}\label{th:pattern-unknown}
Problem~\ref{prob:pattern-unknown} has one-way randomised communication complexity $\Oh (\eps^{-2}\sqrt{n} \log n)$.
\end{theorem}

Having investigated the communication complexity of $(1+\eps)$-approximate Hamming distance we can now define the streaming $(1+\eps)$-approximate Hamming distance problem.

\begin{problem}\label{prob:streaming}
Consider a pattern $P$ of length $n$ and a stream arriving one symbol at a time. We must output a $(1+\eps)$-approximation to the Hamming distance between $P$ and the latest $n$-length suffix of the stream as soon as a new symbol arrives. In this problem we cannot, for example, store a copy of the pattern or stream without accounting for it.
\end{problem}

The upper bounds for the communication complexity of Problem~\ref{prob:pattern-unknown} suggest space upper bounds we shall aim for in order to develop an optimal algorithm for the $(1+\eps)$-approximate Hamming distance in the streaming setting. We make the first step towards this direction and show two randomised sublinear-space algorithms for the problem. We start by giving a solution for the case when both the pattern and the text are binary strings.  

\begin{theorem}\label{thm:streaming}
When both $P$ and $T$ are binary, there is an algorithm for Problem~\ref{prob:streaming} which uses $\Oh(\eps^{-3} \sqrt{n} \log^{2} n)$ bits of space and runs in $\Oh(\eps^{-2} \log{n})$ time per arriving symbol.
\end{theorem}

The same bounds hold for alphabets of constant size $\sigma$ as we can map each occurrence of the $i^{th}$ symbol of the alphabet in the pattern or in the text to a binary string $0^{i-1} 1 0^{\sigma-i}$, which will result in doubling the Hamming distance between the pattern and the text at each particular alignment. 

For polynomial size alphabets our bounds are higher by a factor $\eps^{-2} \log^2 n$ and our approach is based on the mapping idea of Karloff~\cite{Karloff:1993}. In that paper he showed that there exists $\Theta(\eps^{-2} \log^2 n)$ mappings $map_j$ of the alphabet onto $\{0,1\}$ such that an $(1+\eps/3)$-approximation of the Hamming distance between $P$ and $T$ at a particular alignment can be given by a normalised average of the Hamming distances between $map_j(P)$ and $map_j(T)$ at this alignment. Moreover, Karloff showed that the mappings can be generated in $\Oh (\eps^{-2} \log^3 n)$ space and $\Oh (\log n)$ time per symbol. For each pattern-text pair mapped on to a binary alphabet we then run the algorithm of Theorem~\ref{thm:streaming} to find $(1+\eps/3)$-approximations and finally obtain: 

\begin{theorem}
There is an algorithm for Problem~\ref{prob:streaming} which uses $\Oh(\eps^{-5} \sqrt{n} \log^{4} n)$ bits of space and runs in $\Oh(\eps^{-4} \log^3 {n})$ time per arriving symbol.
\end{theorem}

Our solution has guaranteed worst case complexity per arriving symbol and uses roughly the square root of the space required by the known offline $(1+\eps)$-approximate algorithms. A key technical innovation for our space reduction is the notion we introduce of a super-sketch. This a compact and efficiently updateable representation of consecutive text substrings which we require to be able to achieve sublinear space.  For simplicity we will make the natural assumption throughout that  $\eps < 1/2$.

\subsection{Related work and lower bounds}
The one-way communication complexity of a number of variants of Hamming distance computation has been studied over the years. These includes $(1+\eps)$-approximation~\cite{JW:2013}, the so called gap Hamming problem~\cite{CR:2012} and a bounded version known as $k$-mismatch~\cite{HSZZ:06}. However all this previous work has assumed that both Alice and Bob have strings of the same length and so need only give a single output. There has also been great interest in efficient streaming algorithms over the last 20 years, following the seminal work of~\cite{AMS:1996}. In relation specifically to pattern matching problems,  where space is not limited but where an output must be computed after every new symbol of the text arrives, the Hamming distance between the pattern and the latest suffix of the stream can be computed online in $O(\sqrt{n\log{n}})$ worst-case time per arriving symbol or $O(\sqrt{k}\log{k} + \log{n})$ time for the $k$-mismatch version~\cite{CS:2010}.  Both these methods however require $\Theta(n)$ space.  Using the same approach, a number of other approximate pattern matching algorithms have also been transformed into efficient linear space online algorithms including~\cite{AAKLP:2008,AABLLPSV:2009,AALP:2006,AFM:1994, AEE:2006,ACHP:2003,LV:1988kdiff}. In 2013 a small space streaming pattern matching algorithm was shown for parameterised matching~\cite{JPS:2013} and in 2016 for the $k$-mismatch problem~\cite{CFPSS:2016}. The latter $k$-mismatch paper is of particular relevance to our work. In ~\cite{CFPSS:2016} as a part of a space-efficient streaming algorithm for the $k$-mismatch problem, the authors presented a $(1+\eps)$-approximate algorithm with space $\Oh(\eps^{-2} k^2 \log^7 n)$ and running time $\Oh(\eps^{-2} \log^5 n)$ per arriving symbol that returns a $(1+\eps)$-approximation for all alignments of the pattern and text where the Hamming distance is at most $k$. The algorithm we give in this current paper can be seen a generalisation of this work, both removing the requirement for a prespecified threshold $k$ and also using less space when $k \gtrsim n^{1/4}$.

\section{Overview}
\label{sec:overview}
        In this section we give an overview of the main ideas needed for our results. We will make extensive use of sketching.  Alon, Matias and Szegedy were first to show that sketching can be used to approximate frequency statistics of a stream with a particular emphasis on $F_2$~\cite{AMS:1996}. Later their sketching technique was generalized to allow approximation of $||x_1-x_2||_p$ for two vectors $x_1$ and $x_2$ and any $p \in (0; 2]$ by Indyk et al.~\cite{Indyk-stable,CDIM:03}. We will use the sketches of Indyk et al.\@ to show the communication complexity upper bounds. These sketches are based on $p$-stable distributions and have the advantage that they can be used even for large-size alphabets. For our streaming algorithm where we assume that the input alphabet is binary we will use simpler sketches based on the original technique of Alon et al.  

\subsection{Communication complexity}
To show communication complexity bounds we will be using sketches based on $p$-stable distributions (see~\cite{Indyk-stable} and Sections 4.1 and 5.1 of~\cite{CDIM:03}). Setting $\sigma$ to be the alphabet size, a sketch of a string $x$ is defined as a vector $sk (x)$ of length $\Theta(\eps^{-2})$ such that

$$sk (x) [i] = \sum_j Y_{i,j} \cdot x[j]$$

\noindent where each $Y_{i,j}$ is drawn independently from a random stable distribution with parameter $p \le \eps / \log \sigma$. For two vectors $x_1$ and $x_2$ it can be shown that with constant probability the median of values $|sk (x_1) [i] - sk (x_2) [i]|$, appropriately scaled, is a $(1+\eps)$-approximation of the Hamming distance. Importantly, variables $Y_{i,j}$ can be generated when we need them with the help of Nisan's pseudo-random generator, which requires only $\Oh( \log^2 n)$ random bits. 

\subsubsection{Problem~\ref{prob:pattern-known}~--- both Alice and Bob know the pattern}
The main idea of our communication complexity upper bound for Problem~\ref{prob:pattern-known} is that if the Hamming distance between the text and the pattern at a particular alignment is (relatively) small, then Alice and Bob can use the pattern to describe the part of the text aligned with the pattern. 

At each alignment the pattern can be divided into two parts~--- a prefix, aligned with Alice's half of the text, and a suffix, aligned with Bob's half of the text. Alice needs to transmit information that will help Bob approximate the  Hamming distance between these different prefixes of the pattern and her half of the text. She does so by selecting a logarithmic number of prefixes of the pattern with Hamming distances $\Theta(\eps^{-j})$ from the text. She then divides the part of the text aligned with each of these prefixes into blocks such that the mismatches are evenly spread across the blocks, and sends each block's starting position and sketch to Bob.

When Bob wants to compute the Hamming distance between a prefix $P'$ of the pattern and the text and he knows that this Hamming distance is at least $\Theta(\eps^{-(j-1)})$, he uses the prefix $P_j$ with Hamming distance $\Theta(\eps^{-j})$ and the sketches of associated text blocks. The part of Alice's text aligned with $P'$ can be composed of several full blocks and at most one block suffix. Hamming distances between $P'$ and the full blocks can be approximated with the help of the sketches. To approximate the Hamming distance between $P'$ and the suffix of the block, Bob will substitute the suffix with the aligned part of $P_j$. As the number of mismatches between the suffix and $P_j$ is small compared to $\Theta(\eps^{-(j-1)})$, it will give a good approximation of the Hamming distance.  

\subsubsection{Problem~\ref{prob:pattern-unknown}~--- only Alice knows the pattern}
We start by reviewing some notation introduced in~\cite{CFPSS:2016}.

\begin{definition}\label{def:approxperiod}
The $x$-period of a string $S$ of length $n$ is the smallest integer $\ell > 1$ such that the Hamming distance between $S [1, n-\ell]$ and $S [\ell, n]$ is at most $x$. 
\end{definition}

\begin{example}
The $1$-period of a string $babaa$ is $2$.
\end{example}

\begin{definition}
We define the $\ell$-RLE encoding of $S$ to be the ordered set of the run-length encodings of strings $S_i =  S[i]S[\ell+ i]S[2\ell+i]\ldots S[\lfloor (n-i)/\ell \rfloor \cdot \ell + i]$, where $i \in [1,\ell]$. The size of the $\ell$-RLE encoding is the total number of runs in the encodings of strings $S_i$.
\end{definition}

\begin{example}
Let $S= aab  aab aab aab aab aab aac$. The $3$-RLE encoding of $S$ is: the run-length encoding $(a,7)$ of $S_1 = aaaaaaa$, the run-length encoding $(a,7)$ of $S_2 = aaaaaaa$, and the run-length encoding $(b,6)(c,1)$ of $S_3 = bbbbbbc$. The size of the encoding is $1+1+2=4$.
\end{example}

\noindent Note that $\ell$-RLE encoding of $S$ is deterministic and lossless. In~\cite{CFPSS:2016} it was also shown that if $\ell$ is the $x$-period of a string $S$ for some integer $x$, then the size of the $\ell$-RLE encoding is $\Oh(\ell+x)$. Intuitively, this is because each new run in the encoding of $S_i$ corresponds to a mismatch between $S [1, n-\ell]$ and $S [\ell, n]$, and therefore there can be at most $\ell+x$ runs. 

We now explain the idea of the communication protocol for Problem~\ref{prob:pattern-unknown}. Let the block size $B = \sqrt n$ and the threshold $\tau = 2\eps^{-1} \sqrt{n}$. Bob will compute a sketch for each $B^{th}$ suffix of his half of the text and send it to Charlie. Consider a particular alignment of the pattern and of the text. 

\textbf{Case 1: Hamming distance is large.} The pattern can be divided into three parts: a prefix of length at most $B-1$, a middle part aligned with one of the $n/B$ sketched suffixes of Bob's half of the text, and a suffix aligned with Charlie's half of the text. If the Hamming distance at the alignment is larger than $\tau$, then the prefix can be discarded as it will change the Hamming distance by at most $B = (\eps/2) \cdot \tau$. The Hamming distance between the rest of the pattern and the text can be approximated  easily. Charlie has received the sketch of the middle part of the pattern as well as the sketch of the suffix of Alice's half of the text which aligns with it.  Charlie can combine the sketch from Alice's part of the text with the information he has about his half of the text and then compare this sketch to the pattern sketch as required.

\textbf{Case 2: Hamming distance is small.} The main challenge is therefore alignments where the Hamming distance is smaller than $\tau$. If the $(2+\eps) \tau$-period of the pattern is larger than $B$, then there are at most $n/B$ such alignments. 
In this case, Bob can simply send the Hamming distances for all these alignments to Bob. If the period is at most $B$, then Bob will find the first alignment with small Hamming distance and will use the $\ell$-RLE encoding of the pattern and the full list of mismatches to describe the text. Using this description Charlie will be able to fully recover the corresponding suffix of the text and to compute the Hamming distances for all remaining alignments. The only technicality is that Bob does not know Charlie's half of the text and thus will not be able to compute the Hamming distances between the whole pattern and the text. We elaborate on this in Section~\ref{sec:problem2}.

\subsection{A small space streaming algorithm}
In our small space streaming algorithm we will use simpler sketches which provide a $(1+\eps)$-approximation to the Hamming distance between two binary strings of the same length $B$. The method is now folklore but is essentially an application of the technique of the Johnson-Lindenstrauss lemma~\cite{JL:1984}. To do this we create a random  $\eps^{-2} \times B$ matrix $M$ whose entries are from $\{-1,1\}$. The sketch of a string $x$ of length $B$ is then defined to be equal to $Mx$, and it is known that the appropriately scaled square of the $L_2$ norm of the difference of the sketches of two strings gives a $(1+\eps)$-approximation of the Hamming distance between them. 
We will also be using $M$ to define sketches of strings of length $\ell < B$. In this case, we simply append the strings with $(B-\ell)$ zeros and use the method describe above. The original analysis applies here as well. Finally, we will use $M$ to define ``super-sketches'' of strings of length $n-B$. Assume that a string of length $n-B$ is divided into $n/B-1$ non-overlapping blocks of size $B$. A super-sketch is then defined to be a linear combination of the sketches of the blocks. We elaborate more on sketches and super-sketches in Section~\ref{sec:algo-2}.

Now we give a high-level overview of our algorithm. The algorithm starts by preprocessing the pattern $P$. It computes and stores a super-sketch of each $(n-B)$-length substring of $P$. The algorithm then processes the text in non-overlapping blocks of length $B$, computing a sketch for each block. The blocks' sketches can be maintained efficiently as we need to maintain only one sketch at a time. The algorithm also maintains a super-sketch of the last $n/B-1$ blocks. To compute an approximation of the Hamming distance at a particular alignment, the algorithm divides the pattern into three parts: a prefix of length $(B-i)$, a middle part of length $(n - B)$, and a suffix of length $i$, where the middle part is aligned with a block border (see Figure~\ref{fig:algorithm-problem2}). 

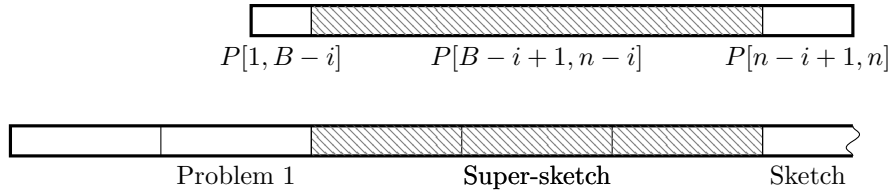
\begin{figure}[h!]
\begin{center}
\begin{tikzpicture}[scale=0.4]
	\usetikzlibrary{patterns}
	
	\draw[very thick] (38,0) -- (10,0)--(10,1)--(38,1);
	\draw[snake=snake] (38,0)--(38,1);
	\draw[very thick] (38,5) rectangle (18,4);

	\foreach \x in {3,4,...,7} {
		\draw (\x*5,0)--(\x*5,1);
	}

	\draw (20,4)--(20,5);
	\draw (35,4)--(35,5);
	
	\node[below] at (19,4) {$P[1,B-i]$};
	\node[below] at (27.5,4) {$P[B-i+1,n-i]$};
	\node[below] at (36.5,4) {$P[n-i+1,n]$};
	
	\draw[pattern=north west lines, pattern color=gray] (20,4) rectangle (35,5);
	\draw[pattern=north west lines, pattern color=gray] (20,0) rectangle (35,1);
	\node[below] at (27.5,0) {Super-sketch};
	
	\node[below] at (17.5,0) {Problem~\ref{prob:pattern-known}};
	\node[below] at (27.5,0) {Super-sketch};
	\node[below] at (36.5,0) {Sketch};

\end{tikzpicture}
\end{center}
\caption{To estimate the Hamming distances at a particular position the algorithm uses a data structure containing the information Alice transfers to Bob in our solution for Problem~\ref{prob:pattern-known} for the prefix $P[1,B-i]$, a super-sketch for the middle part $P[B-i+1, n-i]$, and a sketch for the suffix $P[n-i+1,n]$.}
\label{fig:algorithm-problem2}
\end{figure}

The algorithm then starts by computing the $(1+\eps)$-approximation of the Hamming distance between the middle part and the text with the help of the super-sketches. If the Hamming distance is large, the algorithm can simply discard the prefix and suffix parts. Otherwise, the algorithm also needs to approximate the Hamming distance between the prefix or the suffix of the pattern and the text. To approximate the Hamming distance between the prefix of the pattern and the text the idea is to use the information Alice transfers to Bob in our solution for Problem~\ref{prob:pattern-known}. For the suffix, the algorithm will use the sketch of the part of the block between its start and the current alignment.

\section{Communication complexity}	
In this section we show upper bounds for communication complexities of Problems~\ref{prob:pattern-known} and~\ref{prob:pattern-unknown}. 

\subsection{Problem~\ref{prob:pattern-known}}
\label{sec:problem1}
	We start by showing an upper bound for the communication complexity of Problem~\ref{prob:pattern-known}. Remember that in this problem we have two players, Alice and Bob. Alice knows the first half of the text $T$ and the pattern $P$, and Bob knows the second half of the text $T$ and the pattern $P$. We will show that the communication complexity of this problem is $\Oh(\eps^{-4} \log^2 n)$. 

Let us first explain what Alice sends to Bob. For simplicity, we denote $k = 6 / \eps$. First, Alice selects $q = \lfloor \log_k n \rfloor$ positions $n \ge i_1 \ge i_2 \ge \ldots \ge i_q \ge 1$ such that the Hamming distance between $T[i_j,n]$ and the prefix $P[1,n-i_j+1]$ is at most $k^{j+1}$. She does this in turn starting from $j=1$ and selecting the leftmost possible position for each $j$. Alice then sends to Bob $\Oh(k^2 \cdot \eps^{-2} \log n) = \Oh(\eps^{-4} \log n)$ bits of information for each $j$. She starts by dividing $T[i_j,n]$ into $k^2$ blocks such that the Hamming distance between each block and the corresponding substring of the pattern is at most $k^{j-1}$. If $n \ge b_1 > b_2 > \ldots > b_{k^2} = i_j$ are the borders of the blocks, she sends Bob $b_1, b_2, \ldots, b_{k^2} = i_j$ and the $(1+\eps/6)$-approximate sketches of $T[b_\ell, n]$ for all $\ell \in [1, k^2]$. In total, Alice sends to Bob $\Oh(\eps^{-4} \log^2 n)$ bits of information.

To see how Bob can use this information, consider a particular position~$i$. At this position $P[1, n-i+1]$ is aligned with Alice's half of the text, whereas $P[n-i+2,n]$ is aligned with Bob's half of the text. As Bob knows the pattern, he can compute the exact Hamming distance between $P[n-i+2,n]$ and his half of the text with no additional information. We now go on to explain how he can estimate the Hamming distance $h$ between $P[1, n-i+1]$ and Alice's half of the text. 

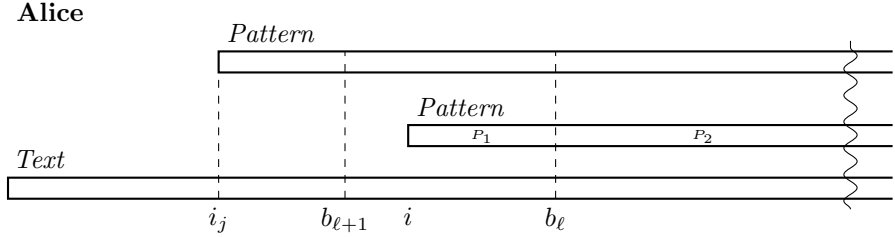
\begin{figure}[h!]
\begin{center}
\begin{tikzpicture}[scale=0.28]
	\draw[thick] (42,0) -- (0,0)--(0,1)--(42,1);
	\draw[thick] (42,2.5) -- (19,2.5)--(19,3.5)--(42,3.5);	
	\draw[thick] (42,6) -- (10,6)--(10,7)--(42,7);
	
	\draw[snake=snake] (40,-0.5)--(40,7.5);	
	
	\node[above] at (2,8) {\textbf{Alice}};	
	\node[above] at (1.5,1) {\textit{Text}};	
	\node[above] at (21.5,3.5) {\textit{Pattern}};	
	\node[above] at (12.5,7) {\textit{Pattern}};
	
	\draw[dashed] (10,0)--(10,6);
	\node[below] at (10,0) {$i_j$};
	\draw[dashed] (16,0)--(16,7);
	\node[below] at (16,0) {$b_{\ell+1}$};	
	\node[below] at (19,0) {$i$};		
	\draw[dashed] (26,0)--(26,7);
	\node[below] at (26,0) {$b_{\ell}$};
	
	\node[above] at (22.5,2.2) {\tiny{$P_1$}};
	\node[above] at (33,2.2) {\tiny{$P_2$}};	
	


\end{tikzpicture}
\end{center}
\caption{Figure shows Alice's half of the text and the rightmost position $i_j < i$. Dashed lines show block borders for $T[i_j, n]$. Borders $b_{\ell+1}$ and $b_{\ell}$ are the closest to $i$ from the left and from the right respectively. The border $b_{\ell}$ divides the pattern into two parts $P_1$ and $P_2$. To estimate the Hamming distance $h_1$ between $P_1$ and $T$, Bob uses the pattern aligned at $i_j$. To estimate the Hamming distance $h_2$ between $P_2$ and $T$, he uses the sketch of $T[b_{\ell},n]$.}
\label{fig:problem1}
\end{figure}

Bob starts by locating the position $i_j$ that is closest to $i$ from the left, and the block $T[b_{\ell+1}, b_{\ell}]$  of $T[i_j,n]$ containing $i$ (see Figure~\ref{fig:problem1}). The border $b_{\ell}$ divides the pattern into two parts, $P_1$ and $P_2$. Let $h_1$ be the Hamming distance between $P_1$ and the text, and $h_2$ be the Hamming distance between $P_2$ and the text, $h_1 + h_2 = h$. To find a $(1+\eps)$-approximation $h'_2$ of $h_2$, Bob uses the sketch of $T[b_\ell, n]$. He cannot use sketches to estimate $h_1$ as $P_1$ is not aligned with the block $T[b_{\ell+1}, b_{\ell}]$, but he knows that there are only a few mismatches between $T[b_\ell, b_{\ell+1}]$ and the pattern aligned at the position $i_j$. So he estimates $h_1$ by computing the Hamming distance $h'_1$ between $P_1$ and the pattern aligned at the position $i_j$. The next lemma shows that Bob can output $h' = (h'_1+h'_2)/(1-\eps/3)$ as a $(1+\eps)$-approximation of $h$.

\begin{lemma}
$h'$ is a $(1+\eps)$-approximation of $h$.
\end{lemma}
\begin{proof}
Remember that $h'_1$ is the Hamming distance between $P_1$ and the pattern aligned at the position $i_j$, and $h_1$ is the Hamming distance between $P_1$ and the text. The Hamming distance between the pattern aligned at the position $i_j$ and $T[b_{\ell+1}, b_{\ell}]$ is at most $k^{j-1}$. Therefore, 

$$h_1 - k^{j-1} \le h'_1 \le h_1 + k^{j-1}$$

\noindent On the other hand, $h'_2$ is a $(1+\eps/6)$-approximation of $h_2$. Hence, 

$$h_1 + h_2 - k^{j-1} \le h'_1 + h'_2 \le h_1 + k^{j-1} + (1+\eps/6) \cdot h_2$$.

\noindent We now substitute $h = h_1 + h_2$ and estimate $h_2 \le h$ to obtain 

$$h - k^{j-1} \le h'_1 + h'_2 \le (1+\eps/6) \cdot h + k^{j-1}$$

\noindent Finally, by our choice of $i_j$ we have $h \ge k^{j+1}$, and therefore

$$(1-\eps/3) \cdot h \le (1-\eps/6) \cdot h \le h'_1 + h'_2 \le (1+\eps/3) \cdot h$$

\noindent Dividing all parts of this inequality by $(1-\eps/3)$, we obtain

$$h \le h' = (h'_1 + h'_2) / (1-\eps/3) \le \frac{1+\eps/3}{1-\eps/3} \; h \le (1+\eps) \cdot h$$
\end{proof}

\subsection{Problem~\ref{prob:pattern-unknown}}
\label{sec:problem2}
	In this section we show an upper bound for the communication complexity of Problem~\ref{prob:pattern-unknown}. Remember that in this problem we have three players, Bob, Charlie, and Alice. Bob knows the first half of the text $T$, Charlie knows the second half of the text $T$, and Alice knows the pattern $P$. We will show that the communication complexity of this problem is $\Oh (\eps^{-2} \sqrt{n} \log n)$. 

Let the block size $B = \sqrt n$ and the threshold $\tau = 2\eps^{-1} \sqrt{n}$. We start by explaining what the players send to each other. Alice sends to Bob the following information:

\begin{enumerate}[1)]
\item $(1+\eps/2)$-approximate sketches of suffixes $P[i, n]$ for all $i \in [1, B]$ (Charlie will use them to estimate large Hamming distances);
\item $(1+\eps/2)$-approximate sketches of prefixes $P[1, n-j B]$ for all $j \in [1, n/B]$ (Bob will use them to find alignments with small Hamming distances);
\item The $\ell$-RLE encoding of the longest prefix $P[1, n-j^\star B]$ with $(2+\eps)\tau$-period $\ell$ smaller than $B$ (Bob will use it to describe the text).
\end{enumerate}

Overall Alice sends $\Oh( (n/B + B) \cdot \eps^{-2} \log n + ((2+\eps)\tau+B) \cdot \log n) = \Oh(\eps^{-2} \sqrt{n}  \log n) $ bits of information.

Bob starts by forwarding the information he received from Alice to Charlie. Bob also sends him $(1+\eps/2)$-approximate sketches of all suffixes $T[jB, n]$. Next, for each $j < j^\star$ Bob uses the sketch of $P [1, n-j B]$ to find $(1+\eps/2)$-approximations of Hamming distances in a block $j$. (Remember that Bob knows $T[1,n]$ and can compute a sketch for any its substring.) If the approximate value of the Hamming distance for some alignment is smaller than $(1+\eps/2) \tau$, he sends it to Charlie. Note that there is at most one such alignment in a block. Indeed, if we have two such alignments in the block, then the Hamming distance between the patterns at these alignments is at most $(2+\eps) \tau$, which is a contradiction with the $(2+\eps) \tau$-period being larger than $B$.
Moreover, Bob will not miss any alignment with the Hamming distance smaller than $\tau$. 

\begin{figure}[h!]
\begin{center}
\begin{tikzpicture}[scale=0.28]
	\draw[thick] (42,0) -- (0,0)--(0,1)--(42,1);
	\draw[thick] (42,4) -- (37,4)--(37,5)--(42,5);
	\draw[thick] (42,8) -- (37,8)--(37,9)--(42,9);
	\draw[thick] (37,8) rectangle (13,9) node[pos=.5] {\tiny{$P[1,n-j^{\star} B]$}};
	\draw[thick] (37,4) rectangle (25,5) node[pos=.5] {\tiny{$P[1,n-j^{\star\star} B]$}};	
		
	\draw[snake=snake] (40,-0.5)--(40,9.5);	
	
	\node[above] at (1,1) {\textit{Text}};	
	\node[above] at (15,9) {\textit{Pattern}};		
	\node[above] at (27,5) {\textit{Pattern}};	

	\node[above] at (1,10) {\textbf{Bob}};		
	
	\foreach \x in {1,...,8,9} {
		\draw (4*\x,0)--(4*\x,1);
	}

	\node[above] at (14,1) {\textit{Block $j^\star$}};
	\node[above] at (26.5,1) {\textit{Block $j^{\star\star}$}};	
	
	\node at (3,0.5) {$\otimes$};
	\node at (25.5,0.5) {$\otimes$};
	\node at (9,0.5) {$\otimes$};
	\node at (27,0.5) {$\times$};
	\node[below] at (25.5,0) {$p$};			
	\node at (31,0.5) {$\times$};	
	\node at (33,0.5) {$\times$};
	\node at (35,0.5) {$\times$};				
\end{tikzpicture}
\end{center}
\caption{The figure shows Bob's half of the text. Crosses show alignments where the Hamming distance is at most $\tau$. $P[1, n-j^\star B]$ is the longest prefix with $\tau(2+\eps)$-period smaller than $B$. Block $j^{\star\star} \ge j^\star$ is the first block containing a cross. Bob sends to Charlie sketches of text suffixes starting at blocks' borders, locations of all encircled crosses, and the last block.}
\end{figure}
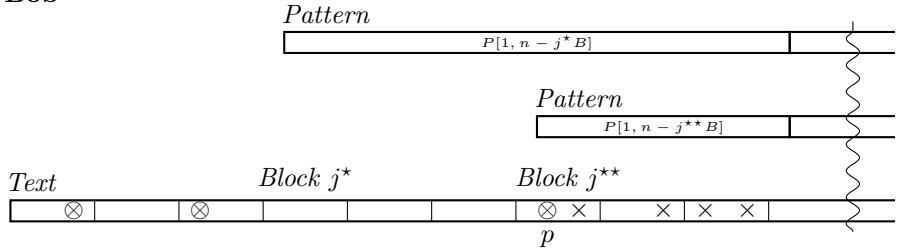

After that, Bob decodes $P[1, n-j^\star B]$ from its $\ell$-RLE encoding and computes the Hamming distances between $P[1, n - j B]$ and the text for all alignments in blocks $j \ge j^\star$ precisely. He finds the first block $j^{\star\star} \ge j^{\star}$ where there is an alignment of $P[1, n - j^{\star\star} B]$ with the Hamming distance at most $\tau$. Bob sends Charlie the starting position~$p$ of this alignment and the positions of the mismatches. Finally, he sends Charlie all bits of the last block of his half of the text. Overall, Bob sends to Charlie $\Oh(\eps^{-2} \sqrt{n} \log n + (\eps^{-2} \log{n} + \log n) \cdot (n / B) + \tau) = \Oh (\eps^{-2} \sqrt{n} \log n)$ bits of information.

We now explain how Charlie computes the Hamming distances. If the Hamming distance at a particular alignment $i$ in a block $j < j^{\star\star}$ is smaller than $\tau$, then Charlie already knows its approximate value. If it is bigger than $\tau$, then Charlie computes its approximation using the sketch of the longest suffix $P[j B-i,n]$ of $P$ aligned with a block border, the sketch of $T[(j+1) B,n]$, and $T[n+1,2n]$. Let $h$ be the Hamming distance between the text and the pattern at the alignment $i$ and let $h'$ be the Hamming distance between $T[(j+1) B, i+n-1]$ and $P[j B-i,n]$.

\begin{lemma}
$h \le h' + B \le (1+\eps) \cdot h$
\end{lemma}
\begin{proof}
The left inequality is trivial. To prove the right one, remember that $\tau \le h$, which implies $B = (\eps/2)\tau \le (\eps/2) \cdot h$. 
\end{proof}

We now go on to the remaining blocks. The Hamming distances at alignments $i < p$ in the block $j^{\star\star}$ are bigger than $\tau$ and Charlie can find their approximation in the way described above. Charlie then decodes $P[1,  n-j^\star B]$ and recovers $T[p,n]$ by fixing the at most $\sqrt{n}$ mismatches between $P[1,  n-j^\star B]$ and $T[p, p + n - j^{\star\star} B + 1]$ and appending the last $p - (j^{\star\star}-1) B$ symbols of $T$ (Remember that Charlie knows all symbols of the last block of $T[1,n]$). Using $T[p, n]$, $T[n+1, 2n]$, and the sketch of $P$, he can approximate the Hamming distances for all alignments to the right of $p$.

\section{Streaming algorithm}
\label{sec:algo-2}
	We now show a streaming algorithm for Problem~\ref{prob:streaming}. In this problem we are asked to output a $(1+\eps)$-approximation of the Hamming distance between the pattern and the text at each alignment, and we do not assume that we store a copy of the pattern or of the text. For $\eps < 1/2$, the algorithm uses $\Oh (\eps^{-3} \sqrt{n} \log^{1.5} n)$ bits of space and its running time is $\Oh(\eps^{-2} \log n)$ per arriving symbol. For simplicity, we will set $k = 1/\eps > 2$ for the rest of this section. 

Let $B = k \sqrt{n}$. The algorithm starts by selecting a $9k^2 \times B$ matrix $M$ and a vector $(\sigma_1, \sigma_2, \ldots, \sigma_{n/B-1})$ of i.u.d.\@  $\pm 1$ random variables. The algorithm then preprocesses the pattern $P$. It remembers the first $B$ symbols of $P$, as well as a super-sketch of each $(n-B)$-length substring of $P$. To compute the super-sketches the algorithm divides a substring into $(n/B-1)$ blocks of length $B$, computes their sketches using $M$ as described in Section~\ref{sec:overview}, and then sums the sketches multiplying them by $\sigma_i$. The algorithm also computes sketches of the last $B$ suffixes of $P$. The sketch of a suffix $P[n-i+1,n]$ is defined to be equal to $M \cdot S_i$, where $S_i = P[n-i+1,n] \; 0^{B-i}$. Finally, for each $i \in [1, B]$ and for each $j \in [0, \log_{1+\eps} n]$ it stores the maximal length of pattern's prefix such that the Hamming distance between this prefix aligned at position $i$ and the pattern is at most $(1+\eps)^j$, which takes $\Oh(\eps^{-1} B \log^2 n)$ bits since $\log_{1+\eps} n = \Oh (\eps^{-1} \log n)$.

\subsection {Text processing}
The algorithm processes the text in non-overlapping blocks of length $B$. For each of the last $n/B$ blocks the algorithm maintains its sketch and a data structure containing the information Alice transfers to Bob in our solution for Problem~\ref{prob:pattern-known}. 

Let us start by explaining how the algorithm maintains the sketches. At the starting index of each block it initialises the block's sketch with a zero vector of length $9 k^2$. When the $j^{th}$ symbol of the block arrives, the algorithm adds the product of the $j^{th}$ column of $M$ and the symbol to the sketch in $\Oh(9k^2)$ time. While reading the block the algorithm also computes the super-sketch of the $(n-B)$-length substring consisting of the $n/B-1$ most recent blocks. Recall that the super-sketch is defined to be equal to the sum of the blocks' sketches multiplied by the variables $\sigma_i$. The total time needed for computing the sum is $\Oh(9k^2 n/ B)$. The algorithm de-amortises this time over the block executing $\Omega(9k^2 n/ B^2)$ steps per arriving symbol.

For each block the algorithm maintains a data structure containing the information Alice transmits to Bob in our solution for Problem~\ref{prob:pattern-known}. The algorithm starts computing the data structure when it has received the entire block. It then computes the Hamming distance between prefixes $P[1], P[1,2], \ldots, P[1,B]$  as being aligned at the right border of the block and the block by running the fast Fourier transform algorithm on $P[1,B]$ and the block appended with $B$ zeros, which takes $\Oh(B)$ space and $\Oh(B \log B)$ time in total~\cite{FP:1974}. The algorithm then finds $i_1, i_2, \ldots, i_q$, where $q = {\lceil \log_k B \rceil}$ as defined in Problem~\ref{prob:pattern-known} and for each $i_j$ it computes the borders and the sketches of the blocks, where the sketches are defined with the help of the matrix $M$. Remember that the algorithm stores the block and the first $B$ symbols of the pattern, so this could be done in a naive way, using symbol-by-symbol comparison. Finally, the algorithm builds binary search trees on $i_1, i_2, \ldots, i_q$ and the block borders for each $i_j$ to allow fast access to the information. The total construction time of the data structure is $\Oh((B + k^2 ) \cdot \log n)$. Note that the data structure will only be used $n/B - 1 \ge 2$ blocks later, so we can de-amortise the construction time over the succeeding block executing $\Omega((1 + k^2 / B) \cdot \log n)$ steps of the construction process per symbol. The data structure occupies $\Oh (k^4 \log^2 n)$ bits of space.

\subsection {Hamming distance}
To compute the Hamming distance at an alignment $i$, the algorithm divides the pattern into three parts: a prefix of length $(B-i)$, a middle part of length $(n - B)$, and a suffix of length $i$, where the middle part is aligned with a block border. The algorithm then starts by computing the square $N$ of the norm of the difference between the super-sketches of the middle part and the corresponding text substring. Both super-sketches are already known as the middle part is an $(n-B)$-length substring of the pattern and we store its super-sketch explicitly, while the super-sketch of the text substring was computed at the end of the preceding block. As both sketches have length $9k^2$, it takes $\Oh(9k^2)$ time. Next, the algorithm computes the Hamming distance $H_s$ between the sketch of the suffix of the pattern and the part of the text block seen so far. This again takes $\Oh(9k^2)$ time. Finally, the algorithm computes an approximation $H_p$ of the Hamming distance between the prefix and the text as described in Problem~\ref{prob:pattern-known}. With the help of the binary search trees, $i_j, b_{\ell+1}$ and $b_{\ell}$ can be found in $\Oh (\log \log n + \log \log k^2)$ time. Recall that $b_{\ell}$ divides the prefix into two parts. The Hamming distance between the second part of the prefix and the text can be approximated in $\Oh (9k^2)$ time with the help of the sketches as in Problem~\ref{prob:pattern-known}, but it is not possible to use symbol-by-symbol comparison for the first part as this would take too much time. Instead, the algorithm does binary search on the prefixes' lengths it calculated during the preprocessing step which allows him to find $(1+\eps)$-approximation of the Hamming distance in $\Oh (\log \log_{1+\eps} n)$ time. It then outputs $H_p + H_m + H_s$, where $H_m = \eps^2 N /9 (1-\eps/3)$. 

\subsection{Analysis}
The running time of the algorithm is $\Oh (\eps^{-2} \log n)$ per arriving symbol. The space used is $\Oh(\eps^{-3} \sqrt{n} \log^2 n)$ bits. We now need to show that $H_p + H_m + H_s$ is a $(1+\eps)$-approximation of the Hamming distance with constant probability. It suffices to show that $H_m$ is a $(1+\eps)$-approximation of the Hamming distance between the middle part of the pattern and the text. Consider two binary strings $t$ and $p$ of length $(n-B)$. Let $sk_t$ and $sk_p$ be their super-sketches of length $9k^2$ calculated with the help of $M$ and $\sigma_i$ and let $N = \norm{sk_t - sk_p}^2$ and $\tilde{H} = \tilde{\eps}^2 N$, where $\tilde{\eps} = \eps / 3$. We will show that $\tilde{H}$ is a good approximation of the Hamming distance between $t$ and $p$. Recall that $t$ and $p$ are binary, and therefore the Hamming distance between them is equal to $\norm{t-p}^2$. 

\begin{lemma}\label{lm:sketches}
With constant probability $(1-\tilde{\eps}) \cdot \norm{t-p}^2 \le \tilde{H} \le (1+\tilde{\eps}) \cdot \norm{t-p}^2$.
\end{lemma}
\begin{proof}
Let $t_i$ and $p_i$, $i \in [1, n/B-1]$, be the blocks of $t$ and $p$ of length $B$. We have 

$$\expect {\tilde{H}} = \tilde{\eps}^2 \cdot \expect {\norm{ \sum_i \sigma_i M \cdot (t_i-p_i)}^2} = \tilde{\eps}^2 \sum_j \expect { \left (\sum_i \sigma_i M_j \cdot (t_i-p_i) \right)^2}$$

\noindent where $M_j$ is the $j^{th}$ row of $M$. As all rows of $M$ are identically distributed, we have $\expect { \left (\sum_i \sigma_i M_j \cdot (t_i-p_i) \right)^2} = \expect { \left (\sum_i \sigma_i M_1 \cdot (t_i-p_i) \right)^2}$ for all $j$, which is equal to $\norm{t-p}^2$ as if at least one of the inequalities $i_1 = i_2$ or $j_1 = j_2$ does not hold, then the variables $\sigma_{i_1} M_1[j_1]$ and $\sigma_{i_2} M_1[j_2]$ are independent and the expectation of $\sigma_{i_1} \sigma_{i_2} M_1[j_1] M_1[j_2]$ is equal to zero, and otherwise it is equal to one. So finally we have $\expect {\tilde{H}} = \norm{t-p}^2$.

\noindent We now compute the variance of $H$. We again use the fact that the rows of $M$ are independent and identically distributed.

$$\var{\tilde{H}} = \tilde{\eps}^2 \cdot \var{\left(\sum_i \sigma_i M_1 \cdot (t_i-p_i) \right)^2} \le \tilde{\eps}^2  \cdot \expect{\left(\sum_i \sigma_i M_1 \cdot (t_i-p_i) \right)^4}$$

\noindent By Khintchine's inequality there exists a universal constant $c > 0$ such that 

$$\var{\tilde{H}} \le c \; \tilde{\eps}^2  \cdot \expect {\left(\sum_i \sigma_i M_1 \cdot (t_i-p_i) \right)^2}^2 \le c \; \tilde{\eps}^2 \cdot \norm{t-p}^4$$

\noindent The claim then follows by Chebyshev's inequality.
\end{proof}

Let now $H = \eps^2 N / 9 (1-\eps/3) = \tilde{H} / (1-\eps/3)$. The probability $H$ is a $(1+\eps)$-approximation of the Hamming distance between $t$ and $p$ is at least the probability $\tilde{H}$ is in $[(1-\eps/3) \cdot \norm{t-p}^2, (1-\eps/3)(1+\eps) \cdot \norm{t-p}^2]$, which in turn can be estimated from below as 

$$\prob{\tilde{H} \in [(1-\eps/3) \cdot \norm{t-p}^2, (1+\eps/3) \cdot \norm{t-p}^2]} \ge 1 - 1/c \mbox{ (Lemma~\ref{lm:sketches}.)}$$

To justify the last transition note that $(1-\eps/3)(1+\eps) \ge  (1+\eps/3)$ for all $\eps < 1$. From above it follows that with constant probabilities $H_p$, $H_m$, and $H_s$ are $(1+\eps)$-approximations of the Hamming distances for the prefix, the middle part, and the suffix of the pattern respectively. We note that the probabilities can be made arbitrarily small by Chebyshev's inequality if we run a constant number of independent instances of the algorithm in parallel and output the sum of the medians of the values $H_p$, $H_m$, $H_s$. Correctness of the algorithm follows by the union bound.

\section{Acknowledgements}
We are grateful to T.S.\@ Jayram and Paul Beame for helpful and inspiring conversations about the problems in this paper and to Ely Porat for introducing the original streaming pattern matching problem to us.

\bibliographystyle{plain}
\bibliography{longnames,bib-latest,bristol}
\end{document}